\documentclass{llncs}

\usepackage{makeidx}
\usepackage[english]{babel}
\usepackage[utf8]{inputenc}
\usepackage[bookmarks]{hyperref}
\usepackage{wrapfig}
\usepackage{amsmath}
\usepackage{amssymb}

\newif\ifappendix
\appendixtrue

\usepackage{listings}
\lstdefinelanguage{pseudo}{
	morekeywords={procedure, if, then, else, while, assume, assert, return},
	sensitive=false,
	morecomment=[l]{//},
	morecomment=[s]{/*}{*/},
	morestring=[b]",
}
\newcommand{\dline}{\noindent\hrulefill}
\newcommand\ignore[1]{{}}

\newcommand{\lt}{{\scriptstyle\mathbf{LT}}} 
\newcommand{\inc}{{\mathrm{IC}}} 
\newcommand{\syc}{{\mathrm{SC}}} 
\newcommand{\gbdiv}[2]{\mathrm{NF}_{#2}(#1)} 
\newcommand{\stem}{{\mathsf{stem}}} 
\newcommand{\exit}{{\mathsf{exit}}} 
\newcommand{\post}{{\mathsf{post}}} 
\newcommand{\reals}{{\mathbb{R}}}
\newcommand{\complex}{{\mathbb{C}}}
\newcommand{\integers}{{\mathbb{Z}}}
\newcommand{\nats}{{\mathbb{N}}}

\newcommand{\pdftitle}{Synthesis for Polynomial Lasso Programs}
\newcommand{\pdfkeywords}{loop invariant, program synthesis,
  termination, Gröbner basis}
\newcommand{\pdfsubject}{Static Analysis}

\hypersetup{
	pdftitle={\pdftitle},
	pdfauthor={Jan Leike, Ashish Tiwari},
	pdfkeywords={\pdfkeywords},
	pdfsubject={\pdfsubject},
	pdflang={en-US}
}
\title{\pdftitle\thanks{The final publication is available
at \href{http://link.springer.com/}{link.springer.com}.}}
\titlerunning{\pdftitle}
\author{Jan Leike\inst{1} \and Ashish Tiwari\inst{2}\thanks{Supported in part by DARPA under contract FA8750-12-C-0284 and by the NSF grant SHF:CSR-1017483.}}
\authorrunning{J. Leike and A. Tiwari}
\institute{
  University of Freiburg, Germany.\ \url{leike@informatik.uni-freiburg.de} \\
  \and SRI International, Menlo Park, CA.\ \url{ashish.tiwari@sri.com}
}
\date{\today}

\begin{document}

\maketitle

\begin{abstract}
We present a method for the synthesis of polynomial lasso programs.
These programs consist of a program stem, a set of transitions,
and an exit condition, all in the form of algebraic assertions (conjunctions of
polynomial equalities).  Central to this approach is the discovery of
non-linear (algebraic) loop invariants.  We extend Sankaranarayanan,
Sipma, and Manna's template-based approach and prove a completeness
criterion.  We perform program synthesis by generating a constraint
whose solution is a synthesized program together with a loop invariant
that proves the program's correctness.  This constraint is non-linear
and is passed to an SMT solver.  Moreover, we can enforce the
termination of the synthesized program with the support of test cases.
\end{abstract}


\section{Introduction}
\label{sec-induction}

There have been significant advances in automating 
program verification, and even extending the verification
techniques to perform automated synthesis of correct programs.
Often, automation is achieved using appropriate abstract
domains for analysis.  The choice of abstract domains is governed
by the class of program fragments being analyzed.
In this paper, we are interested
in programs that perform some numerical computation.  For reasoning
about such programs, the theory of polynomial ideals has proven to be
an excellent abstract domain because of two reasons.  First, there is
a nice correspondence between subsets of the program state space and
polynomial ideals (as established in the field of algebraic geometry),
and second, there are effective algorithms for computing with
polynomial ideals.  In this paper, we will use the abstract domain of
polynomial ideals for reasoning about polynomial lasso programs.

In our terminology, a polynomial lasso program consists of an assertion
describing program states before loop entry, an assertion describing
program states after loop termination and a set of transitions
corresponding to the branches in the loop body.  All involved
assertions are algebraic; that is, conjunctions of polynomial
equalities.

Our approach for analysis of such polynomial lasso programs is not based on
iterative fixpoint computation.  Instead, we use the constraint-based
approach, also known as template-based approach, for directly finding
fixpoints using constraint solving.  This way we avoid convergence
issues of iterative fixpoint methods.
Our starting point is a 
method presented by Sankaranarayanan,
Sipma and Manna~\cite{Sankaranarayanan04}.  
Despite its obvious incompleteness, the method is often successful
in verifying programs.  Why is this method ``complete in practice''?
We answer the question here by presenting a first 
completeness criterion for this method.  For this purpose, we have to
extend the original invariance criteria in~\cite{Sankaranarayanan04}
and generate a new and refined \emph{invariance condition}.

Our interest here is not just on the verification problem, but
also on the synthesis problem.  Specifically, taking inspiration from
recent work on synthesis of programs by completing partial program
``sketches''~\cite{sketch,bitvector}, we start with a polynomial lasso program
that contains parameters (variables to be synthesized) and a post
condition.  The goal is to find values for the parameters that result
in a correct program.  We solve the synthesis problem by generating a
\emph{synthesis constraint}---a constraint whose solution provides a
valuation for the parameters.  Additionally, the constraint's
solution also supplies values that define an inductive loop invariant
for the synthesized polynomial lasso program.  This invariant constitutes as
proof that the synthesized program is in fact correct with respect to
the given post condition.  Thus, we simultaneously synthesize the
program and its proof of correctness.  There is one caveat though: if
variables that are critical to termination have parameterized
updates, then the synthesized lasso program might not be terminating.  To solve
this problem, we use a finite number of test cases that specify input
variable assignment, output variable assignment and a sequence of loop
transitions.  These test cases are used to strengthen the synthesis
constraint so that the undesirable solutions are eliminated.

The template-based approach reduces the synthesis problem and the loop
invariant discovery problem into an $\exists\forall$ constraint: the
template variables and the synthesis variables are existentially
($\exists$) quantified, whereas the program variables are universally
($\forall$) quantified~\cite{bitvector}.  We use the theory of
polynomial ideals to (conservatively) eliminate the inner $\forall$
quantifier.  The resulting formula is our synthesis constraint -- an
(existentially quantified) conjunction of non-linear algebraic
equalities -- which is solved by an off-the-shelf non-linear SMT
solver.

We demonstrate that the template-based approach on polynomial ideals
abstract domain can be used to successfully synthesize polynomial
lasso programs.  However, the approach has certain limitations.  First, it
cannot handle inequalities.  Polynomial ideals logically correspond to
conjunctions of polynomial equalities.  Now, inequalities can be
encoded as equalities, but algorithms on polynomial ideals (that
compute canonical Gröbner basis) do not lift easily to reasoning
about the encoded inequalities~\cite{Tiwari05:CSL}.  For handling
inequalities, one could use semialgebraic sets as the abstract domain,
and then use algorithms based on either \emph{cylindric algebraic
  decomposition}~\cite{Collins75} or the
\emph{Positivstellensatz}~\cite{Stengle74,Parrilo03,Tiwari05:CSL}, but
we leave that for future work.

A second issue is the size of the synthesis constraint.  Non-linear
solvers scale very poorly with increasing number of variables and the
synthesis constraint (generated by the synthesis process) can be large
and tends to be non-linear.

The final issue is related to the completeness of our approach.
Incompleteness arises due to the use of templates, and also due to the
use of polynomial ideal theory rather than the theory of reals.  We
address the latter issue in \autoref{sec-invariants}.  For the former
issue, we just have to use polynomial templates with sufficiently
large degree bounds.  In our examples, a general template of degree
two or three was sufficient, but the size of generic template
polynomials grows exponentially with their degree.


\section{Related Work}
\label{sec-related-work}

The automatic discovery of polynomial invariants for imperative
programs has received a lot of attention in recent years.  Müller-Olm
and Seidl generate invariant polynomial equalities of bounded degree
by backwards propagation~\cite{MOS04}.  This can be seen as an
extension to Karr's algorithm~\cite{Karr76}, which uses only linear
arithmetic.  Seidl, Flexeder and Petter apply the
backwards-propagation method to programs over machine integers, i.e.,
programs whose variables range over the domain
$\integers_{2^w}$~\cite{Seidl08}.

Rodrígues-Carbonell and Kapur use an iterative approach based on
forward propagation and fixed point computation on Gröbner bases over
the lattice of ideals to generate the ideal of all loop
invariants~\cite{RCK04,RCK05}.

Colón combines the two aforementioned approaches by doing the fixed
point computation on ideals with linear algebra~\cite{Colon07}.  He
introduces the notion of pseudo-ideals to ensure termination of
the fixed point computation while retaining the expressiveness of
generated invariants.

Polynomial program invariants can also be derived without using Gröbner
basis computations~\cite{CJJK12}. Cachera et al.\ use backwards analysis and
variable substitution on template polynomials for an incomplete approach.

The constraint solving approach that generates invariant polynomial
equalities using templates was proposed by Sankaranarayanan, Sipma and
Manna~\cite{Sankaranarayanan04}.
Invariant generation is a central ingredient to our synthesis method, so we want the invariant generation process to be as complete as possible.
Therefore we extend their approach by using a more general condition for the invariant (see also \autoref{rem-ssm-comparison}) that
enables us to state a completeness criterion.

Polynomial lasso programs have also received some attention regarding the
analysis of their termination properties.  Bradley, Manna and Sipma
use finite difference arithmetic to compute lexicographic polynomial
ranking functions for polynomial lasso programs~\cite{Bradley05}.

All the aforementioned papers consider the verification (or the
invariant generation) problem.  In this paper, inspired by recent work
on program synthesis~\cite{sketch}, we also consider the synthesis
problem.
Our work can be considered a more formal approach to Col{\'o}n's
method~\cite{Colon04} that uses non-linear constraint solving to
instantiate program schemata (parameterized programs augmented with
constraints).
Our approach relies on algebraic methods instead of heuristics.

Finally, Srivastava et al.~\cite{SSG10} describe a big-picture program synthesis algorithm from scaffolds.
These scaffolds consist of pre- and postconditions, a program flow template, and bounds on the number of variables and the number of local branches.
For the synthesis condition, all control flows of the template program are unfolded and constraints are generated with respect to invariants and ranking functions ensuring the program's correctness and termination.
This constraint is then proven by a specialized external method and
our algorithm can be used as one of these external methods.


\section{Preliminaries}
\label{sec-preliminaries}

Let $V$ be a set of variables, $V = \{ x_1, \ldots, x_n \}$.  The
variables of the `next state' are denoted by the corresponding primed
variables $V' = \{ x_1', \ldots, x_n' \}$.  Having both primed and
unprimed variables in an expression enables stating a relationship
between two states.

For the set of real numbers $\reals$, let $\reals[V]$ denote
the ring of polynomials in the variables $V$ with coefficients from
$\reals$.  
A subset $I \subseteq \reals[V]$ is an \emph{ideal} if
(a) $0 \in I$, 
(b) $f + g \in I$ for all $f, g \in I$, and
(c) $h \cdot f \in I$ for all $f \in I$ and $h \in \reals[V]$.
For a set of polynomials $P = \{ p_1(V), \ldots, p_k(V) \}$,
{\em{the ideal $\left< P\right>$ generated by $P$}} is 
\[
\left< P \right> = \left< p_1, \ldots, p_k \right>
= \Big\{ \sum_{i=1}^k q_i(V) p_i(V) \,\Big|\, q_1, \ldots, q_k \in
\reals[V] \Big\}.
\]
Note that if all polynomials in $P$ evaluate to $0$ at any point in $\reals^n$,
then all polynomials in $\left< P\right>$ will also evaluate to
$0$ at that point.

By the Hilbert Basis Theorem, every ideal $I$ has a finite set of
generators.  
Moreover, for a fixed ordering on the monomials (such as
total degree lexicographic ordering induced by any precedence
relation on the variables),
there is a finite ``canonical'' set of
generators of $I$ called a \emph{Gröbner basis}.  
A Gröbner basis $G =
\{ g_1, \ldots, g_k \}$ for $I$ has the following properties~\cite{Cox91}.
\begin{enumerate}
  \item $G$ is computable in {\small{DOUBLE-EXPSPACE}} from a set of generators of $I$
  (Buchberger's Algorithm).
\item For all $p \in \reals[V]$, the result of division of $p$ on
  $G$, denoted $\gbdiv{p}{G}$, is unique and does not depend on the
  order in which the division steps are performed. 
\item For all $p \in \reals[V]$, $\gbdiv{p}{G} = 0$ iff $p \in I$.
\end{enumerate}
For example, if $P = \{xy-2, x^2-4\}$ and we use the
precedence $x \succ y$, then
$G = \{x-2y, y^2-1\}$ is a 
Gröbner basis for the ideal $\left< P\right>$.
Division of $p$ on $G$ can be performed by replacing $x$ by $2y$
and replacing $y^2$ by $1$ in $p$ repeatedly.
The result $\gbdiv{x^2+y^2-5}{G}$ of division of $x^2+y^2-5$ on $G$
is $0$, and hence we can conclude that $x^2+y^2-5\in\left< P\right>$.

\begin{definition}[Radical Ideal]\label{def-radical}
An ideal $I$ is a \emph{radical ideal} if $f^m \in I$ implies $f \in
I$ for every $m \in \nats$.
\end{definition}
Given an ideal $I$, note that the set 
$\{f \mid \exists{m \in\nats}: f^m\in I\}$ is a (radical) ideal.

\begin{definition}[Algebraic Assertion]\label{def-algebraic-assertion}
An \emph{algebraic assertion} $\varphi(V)$ (or just $\varphi$) 
over the set of variables $V$ is a
formula of the form $\bigwedge_{i=1}^m p_i(V) = 0$ where each $p_i \in
\reals[V]$ for $1 \leq i \leq m$.
\end{definition}

An algebraic assertion $\bigwedge_{i=1}^m p_i(V) = 0$ generates an
ideal $\left< \varphi\right> = \left< p_1, \ldots, p_m \right>$.  
We will use $\varphi$ to denote the formula as well as the
set of polynomials $\{p_1,\ldots,p_m\}$ in the formula.
An assertion $\varphi$ can be interpreted in the theory $\reals$
of reals or in the theory $\complex$ of complex numbers. 
A {\em{valuation}} is a mapping from variables to values (in the set
of real numbers or the set of complex numbers).
A polynomial in $\reals[V]$ evaluates to a value 
(in $\reals$ or $\complex$) for a given valuation for $V$.

\begin{theorem}[Zero Polynomial Theorem]\label{thm-polynomial-zero}
A polynomial $p \in \reals[V]$ is zero for all possible valuations
$\nu: V \to \reals$ if and only if all of its coefficients are
zero.
\end{theorem}

\begin{lemma}\label{lem-nullstellensatz-converse}
Let $\varphi$ be an algebraic assertion over $V$ and $p \in
\reals[V]$ a polynomial.  If $p \in \left<\varphi\right>$, 
then $\reals
\models \varphi(V) \rightarrow p(V) = 0$.
\end{lemma}

\begin{theorem}[Hilbert's Nullstellensatz~\cite{Cox91}]\label{nullstellensatz}
Let $\varphi$ be an algebraic assertion and $p \in \complex[V]$ a
polynomial.  If $\left<\varphi\right>$ is a radical ideal and $\complex \models
\varphi(V) \rightarrow p(V) = 0$, then $p \in \left<\varphi\right>$.
\end{theorem}

\begin{lemma}\label{lem-ideal}
Let $p, s \in \reals[V]$ and $\left<\varphi\right> \subseteq \reals[V]$ be an ideal.  Then,
$p \in \left< s, \varphi \right>$ if and only if
there is a polynomial $t \in
\reals[V]$ such that $p - t \cdot s \in \left<\varphi\right>$.
\end{lemma}
\begin{proof}
Let $\left<\varphi\right> = \left< p_1, \ldots, p_k \right>$.  By definition, $p
\in \left< s, \varphi \right>$ iff there are $t, t_1, \ldots, t_k \in
\reals[V]$ such that $p = ts + \sum_i t_i p_i$.  This is
equivalent to $p - ts = \sum_i t_i p_i$, which holds iff $p - ts \in
\left<\varphi\right>$. \qed
\end{proof}

Similar to the definition by Sankaranarayanan et al., we introduce
template polynomials as a means for finding polynomials with certain
properties.  In our definition the template coefficients can be
non-linear polynomials. For the mathematical details regarding
template polynomials, see \cite{Sankaranarayanan04}.

\begin{definition}[Template Polynomial]\label{def-template}
Let $A$ and $V$ be two disjoint sets of variables.
A \emph{template polynomial} or \emph{template} over $(A, V)$ is a
polynomial with variables $V$ and coefficients from $\reals[A]$.
A template is said to be a \emph{linear template} if all of its
coefficient polynomials are linear.
\end{definition}

Template polynomials will be denoted by upper case Greek letters.
Given a degree bound $d$, the \emph{generic template polynomial}
$\Psi$ over $(A, V)$ of total degree $d$ is given by
\[
\Psi(V) = \sum_{|\gamma| \leq d} a_\gamma V^\gamma
\]
where $\gamma \in \nats^{\#V}$ is a multi-index and $A = \{
a_\gamma \,|\, \gamma \in \nats^{\#V} \}$ are template variables.

\begin{definition}[Semantics of Templates]\label{def-template-semantics}
For a set of template variables $A$, an \emph{$A$-valuation} is a map
$\alpha: A \to \reals$.  This map can be naturally extended to a map
$\tilde\alpha: \reals[A][V] \to \reals[V]$ that replaces every
occurrence of an $a \in A$ by $\alpha(a)$.
\end{definition}


\section{Polynomial Lasso Programs}
\label{sec-lasso}

We define the syntax and semantics of polynomial lasso programs.
We also define inductive invariants for such programs.
Henceforth, semantic entailment, $\models$, should always be
interpreted as in the theory $\reals$ of reals.

\begin{definition}[Polynomial Lasso Program]\label{def-lasso}
A \emph{polynomial lasso program} $L = (V, \stem, \mathcal{T}, \exit)$ consists
of
\begin{itemize}
\item a set of variables $V$,
\item an algebraic assertion $\stem$ over $V$ called the \emph{program stem},
\item a \emph{set of transitions} $\mathcal{T}$, where each transition
  $\tau \in \mathcal{T}$ is an algebraic assertion over $V \cup V'$,
\item and an algebraic assertion $\exit$ over $V$, called the
  \emph{exit condition}.
\end{itemize}
A transition $\tau$ is said to be \emph{deterministic} if it can be
written in the form
\[
\bigwedge_j h_j(V) = 0 \;\land\; \bigwedge_i x_i' g_i(V) - f_i(V) = 0,
\]
where every $x_i' \in V'$ occurs exactly once and $\neg \exit \models
g_i(V) \neq 0$.
For every $i$ and $j$, the polynomial $h_j$ is called \emph{guard} and the
polynomial $x_i' g_i(V) - f_i(V)$ is called \emph{update}: $f_i$ is
its \emph{numerator} and $g_i$ its \emph{denominator}.
The polynomial lasso program $L$ is called \emph{pseudo-deterministic} if all
its transitions $\tau \in \mathcal{T}$ are deterministic.
\end{definition}

Lassos with solely deterministic transitions can have overlapping
guards, hence the choice of transitions may be non-deterministic even
in a pseudo-deterministic polynomial lasso program.
Due to the nature of imperative languages, pseudo-deterministic
lassos possess a specific interest to us.

\begin{definition}[Semantics of a Lasso Program]\label{def-lasso-semantics}
Let $L = (V, \stem, \mathcal{T}, \exit)$ be a polynomial lasso program.
An \emph{execution of $L$} is a (potentially infinite) sequence $\sigma =
\nu_0 \nu_1 \ldots$ where $\nu_i: V \to \reals$ is a valuation on
the variables $V$ such that
\begin{enumerate}
\item $\nu_0 \models \stem$
\item For all $i \geq 0$ there is a $\tau \in \mathcal{T}$ such that
  $\tau(\nu_i, \nu_{i+1})$.
\item $\nu_i \models \exit$ iff it is the last element in $\sigma$.
\end{enumerate}
\end{definition}

\begin{example}[Running example]\label{ex-product1}
Consider the imperative program and its lasso representation
$L$ shown in \autoref{fig-product}.
$L$ is a pseudo-deterministic lasso program since $\tau$ is a deterministic
transition with the two update polynomials $y' - y + 1$ and $s' - s -
x_0$ and no guards.
An execution of $L$ is $\sigma = \nu_0 \nu_1$ where
\begin{center}
\begin{tabular}{lllll}
$\nu_0$:\;\; & $x_0 \mapsto 3$ & $y_0 \mapsto 1$ & $y \mapsto 1$ & $s \mapsto 0$, \\
$\nu_1$:\;\; & $x_0 \mapsto 3$ & $y_0 \mapsto 1$ & $y \mapsto 0$ & $s \mapsto 3$.
\end{tabular}
\end{center}
\end{example}

\begin{figure}[t]
\begin{center}
\begin{tabular}{c@{$\quad$}|@{$\quad$}c}
\begin{minipage}{50mm}
\begin{lstlisting}
procedure product($x_0$, $y_0$):
    $s$ := $0$;
    $y$ := $y_0$;
    while ($y \neq 0$):
        $s$ := $s + x_0$;
        $y$ := $y - 1$;
    return $s$;
\end{lstlisting}
\end{minipage}
&
\begin{minipage}{60mm}
  Lasso program
$L = (V, \stem, \mathcal{T}, \exit)$: 
\begin{eqnarray*}
  V       &=&    \{ x_0, y_0, y, s \}, \\
\stem &\equiv &  s = 0 \land y = y_0, \\
\tau  &\equiv & y' = y - 1 \land s' = s + x_0, \\
\exit &\equiv & y = 0, \\
\mathcal{T} &=& \{ \tau \}
\end{eqnarray*}
\end{minipage}
\end{tabular}
\end{center}
\caption{An example imperative code and its representation as 
  a polynomial lasso program (see \autoref{ex-product1}).  
  The program performs a multiplication by repeated addition.}
\label{fig-product}
\end{figure}

\begin{definition}[Correctness]\label{def-correctness}
Let $L = (V, \stem, \mathcal{T}, \exit)$ be a polynomial lasso program and
let $\post$ be an algebraic assertion over $V$.
The lasso $L$ is said to be \emph{(partially) correct with respect to
  the post condition $\post$} if for every finite execution $\sigma$
of $L$, the last valuation in $\sigma$ is a model of $\post$.  $L$ is
\emph{totally correct with respect to $\post$} if it is partially
correct with respect to $\post$ and it is terminating, i.e., there are
no infinite executions of $L$.
\end{definition}

\begin{definition}[Invariant]\label{def-invariant}
Let $L = (V, \stem, \mathcal{T}, \exit)$ be a polynomial lasso program.
A polynomial $p \in \reals[V]$ is called an \emph{(inductive)
invariant} of a transition $\tau \in \mathcal{T}$ if
\begin{enumerate}
\item $\stem \models p(V) = 0$ and
\item $p(V) = 0 \land \tau(V, V') \land \neg\exit \models p(V') = 0$.
\end{enumerate}
The polynomial $p$ is called an \emph{(inductive) invariant of $L$} if
it is an invariant of all transitions $\tau \in \mathcal{T}$.
\end{definition}

It is easily shown by means of induction that if $p$ is an invariant of a
lasso $L$, then for every execution $\sigma$ of $L$ and every
valuation $\nu \in \sigma$, we have $\nu \models p = 0$.

\begin{example}\label{ex-product2}
\autoref{ex-product1} calculates the product $s$ of the two input
values $x_0$ and $y_0$ by repeated addition. The polynomial lasso program $L$
is partially correct with respect to the post condition $s = x_0 y_0$
and it is easy to check that $s + x_0 y - x_0 y_0 = 0$ is an invariant
of $L$.
\end{example}


\section{Polynomial Loop Invariants}
\label{sec-invariants}

In this section, we extend
the approach for discovering loop invariants for polynomial lasso programs
introduced by Sankaranarayanan, Sipma and
Manna~\cite{Sankaranarayanan04}. We define a
weakened form of what they call \emph{polynomial consecution}.  We
prove that under some restrictions, this is a complete approach for
invariants over the complex numbers.  The results established in this
section will then be applied to program synthesis in
\autoref{sec-synthesis}.

The first lemma relieves us in certain cases from the potentially very
expensive computation of a Gröbner basis for the loop transitions.  
Specifically, for a
deterministic transition $\tau$, division by the Gröbner basis of 
$\tau$ is equivalent to substitution of the primed variables
according to the update statements.
\begin{lemma}\label{lem-deterministic-lasso}
Let $\tau$ be a deterministic transition with at most one guard
polynomial $h$ and updates $x_i' - f_i(V)$ that have denominator $1$.
If $x_i' \succ x_j$ in the monomial ordering for all $i$ and $j$, then
the set $G = \{ h(V) \} \cup \{ x_i' - f_i(V) \,|\, 1 \leq i \leq n
\}$ is a Gröbner basis of the ideal $\left<\tau\right>$.
\end{lemma}
\ignore{
\begin{proof}
A set of polynomials $G = \{ p_1, \ldots, p_n \}$ is a Gröbner basis
iff for all $i \neq j$, $\gbdiv{S(p_i, p_j)}{G} = 0$.
\begin{align*}
S(x_i' - f_i, x_j' - f_j)
= x_i' f_j - x_j' f_i
\rightarrow_G - x_j' f_i + f_i f_j
\rightarrow_G 0.
\end{align*}
Let $h = \lt(h) + h'$, then
\begin{align*}
S(x_i' - f_i, h)
  = - x_i' h' - \lt(h) f_i
  \rightarrow_G - f_i h' - \lt(h) f_i
  = - h f_i
\rightarrow_G 0. \quad\qed
\end{align*}
\end{proof}
\endignore}

For the remainder of this paper, let $L = (V, \stem, \mathcal{T},
\exit)$ be a fixed pseudo-deterministic polynomial lasso program.
We will now define a sufficient, and under some assumptions also
necessary, condition for a template polynomial to be an invariant of $L$.

\begin{definition}[Invariance Condition]\label{def-invariance-condition}
For each transition $\tau \in \mathcal{T}$, let $q_\tau$ be any 
common multiple of the denominators of the update statements of
$\tau$. 
(In particular, $q_\tau$ can be the product of all denominators.)
Let $\Psi$ be a template polynomial over $(A,
V)$ of total degree $d$.  Let $s(V)$ be the generator of $\exit$ if it
has only one generator and $1$ otherwise.
The \emph{invariance
condition $\inc(L,\Psi)$ of $L$ for $\Psi$}
is the conjunction of
\begin{align*}
\gbdiv{\Psi(V)}{\stem} &= 0, \\
\gbdiv{q_\tau(V)^d \cdot s(V) \cdot \Psi(V')}{\tau} &= \Phi_\tau(V)
\cdot \Psi(V), \text{ for all } \tau \in \mathcal{T},
\end{align*}
where
the polynomials 
$\Phi_\tau$
are generic template polynomials over $(B_\tau, V)$ whose degrees
are bounded by the result of the division $\gbdiv{q_\tau(V)^d \cdot
  s(V) \cdot \Psi(V')}{\tau}$ and
$B_\tau$ are new disjoint sets of template variables.
\end{definition}

The variables $V$ and $V'$ are universally quantified in the
invariance condition, whereas the variables $A$ and $(B_\tau)_{\tau
  \in \mathcal{T}}$ are existentially quantified.  By the Zero
Polynomial Theorem \ref{thm-polynomial-zero}, the equations in the
invariance condition hold for all valuations on $V \cup V'$ if and
only if all the coefficients of the polynomials are identical to zero.
Therefore the variables $V$ and $V'$ can be removed from the
invariance condition yielding a constraint on the variables $A$ and
$(B_\tau)_{\tau \in \mathcal{T}}$.

\begin{remark}\label{rem-ic-components}
The invariance condition is designed to allow completeness in a wide variety of cases.
We provide some intuition for its components below, 
but for details
the reader is referred to the proof of \autoref{thm-completeness}.
\begin{itemize}
\item The result of the division $\gbdiv{q_\tau(V)^d \cdot s(V) \cdot \Psi(V')}{\tau}$ may not yield $\Psi(V)$, but rather some multiple
  of $\Psi(V)$. Hence, we have the generic template
  polynomial $\Phi_\tau$ in the invariance condition.
\item If an update statement, say $x_i'g_i-f_i$, in $\tau$ contains 
  a nontrivial denominator $g_i$, then 
  we may not be able to remove $x_i'$ from $\Psi(V')$ 
  by division on $\tau$.
  Since every monomial in $\Psi(V')$ contains at most $d$ primed variables,
  therefore multiplying $\Psi(V')$ with the polynomial $q_\tau(V)^d$ guarantees that division by $\tau$ will eliminate all primed variables.
\item When the exit condition $s(V) = 0$ holds, we do not need $\Psi$ to be inductive.
 Hence, we use the product $\Psi(V') \cdot s(V)$, which encodes that $\Psi$ holds in the next state \emph{or} the exit condition is satisfied.
\item If the exit condition is generated by more than one polynomial,
we cannot use this trick for all generators, thus loosing completeness.
For simplicity, we set $s = 1$ in those cases, but selecting one of the exit condition's generators as $s$ will make the condition more complete (but also more complex).
\end{itemize}
\end{remark}

\begin{remark}\label{rem-ssm-comparison}
The invariance condition in \autoref{def-invariance-condition} is more general than the condition used by Sankaranarayanan et al.~\cite{Sankaranarayanan04}.
They use the following inductiveness property:
\[
\gbdiv{\Psi(V')}{\tau} - \lambda \cdot \gbdiv{\Psi(V)}{\tau} = 0,
\]
where $\lambda$ is a real-valued variable.
This not only restricts $\Phi_\tau$ to a template of degree $0$, it also omits
the additions we have discussed in \autoref{rem-ic-components}.
\end{remark}

\begin{example}\label{ex-product3}
In order to state the invariance condition for \autoref{ex-product1},
we first fix a template polynomial $\Psi$ over $V$.  The general
second-degree template polynomial over $V$ is the following.
\begin{align*}
\Psi(V) &= a_0 x_0^2 + a_1 y_0^2 + a_2 y^2 + a_3 s^2
+ a_4 x_0 y_0 + a_5 x_0 y + a_6 x_0 s \\
&+ a_7 y_0 y + a_8 y_0 s + a_9 y s
+ a_{10} x_0 + a_{11} y_0 + a_{12} y + a_{13} s + a_{14}
\end{align*}
The invariance condition $\inc(L, \Psi)$ is given by the following
equations.
\begin{align*}
0 &= a_0 x_0^2 + (a_1 + a_2 + a_7) y^2 + (a_4 + a_5) x_0 y + a_{10} x_0
   + (a_{11} + a_{12}) y + a_{14} \\
0 &= (a_0 + a_3 + a_6 - b a_0) x_0^2 y + (a_1 - b a_1) y_0^2 y
   + (a_2 - b a_2) y^3 + (a_3 - b a_3) y s^2 \\
  &+ (a_4 + a_8 - b a_4) x_0 y_0 y + (a_5 + a_9 - b a_5) x_0 y^2
   + (a_6 + 2a_3 - b a_6) x_0 s y \\
  &+ (a_7 - b a_7) y_0 y^2 + (a_8 - b a_8) y_0 y s + (a_9 - b a_9) y^2 s \\
  &+ (a_{10} + a_{13} - a_9 - a_5 - b a_{10}) x_0 y
   + (a_{11} - a_7 - b a_{11}) y_0 y \\
  &+ (a_{12} - 2 a_2 - b a_{12}) y^2 + (a_{13} - a_9 - b a_{13}) y s
   + (a_{14} - b a_{14}) y
\end{align*}
Here, $\Phi_\tau(V) = b \cdot y$ is the generic template polynomial over
$B_\tau = \{ b \}$ of degree $0$ multiplied with $y$, the generator of
$\exit$ (for simplicity of presentation, we abstained from using a generic template polynomial for $\Phi_\tau$).
By \autoref{thm-polynomial-zero}, these two equalities yield
21 equations which are linear after assigning a value to $b$.
The assignment $\alpha: A \cup B_\tau \to \reals$ given by the
following table is a solution to the invariance condition $\inc(L,
\Psi)$.
\begin{center}
\begin{tabular}{r@{\hskip 2mm}|@{\hskip 2mm}ccccccccccccccccc}
         & $b$ & $a_0$ & $a_1$ & $a_2$ & $a_3$ & $a_4$ & $a_5$ & $a_6$ & $a_7$
& $a_8$ & $a_9$ & $a_{10}$ & $a_{11}$ & $a_{12}$ & $a_{13}$ & $a_{14}$ \\
$\alpha$ & $1$ & $0$ & $0$ & $0$ & $0$ & $-1$ & $1$ & $0$ & $0$ & $0$ & $0$
& $0$ & $0$ & $0$ & $1$ & $0$
\end{tabular}
\end{center}
This yields the loop invariant $\tilde\alpha(\Psi) = s + x_0 y - x_0 y_0$
from \autoref{ex-product2}.
\end{example}

\begin{theorem}[Soundness]\label{thm-soundness}
If $\alpha: A \cup \bigcup_{\tau \in \mathcal{T}} B \to \reals$ is
an assignment for the template variables that is a solution to the
invariance condition $\inc(L, \Psi)$, then $\tilde\alpha(\Psi)$ is an
invariant of $L$.
\end{theorem}
\begin{proof}
$\gbdiv{\tilde\alpha(\Psi)}{\stem} = 0$, hence $\tilde\alpha(\Psi) \in
\left<\stem\right>$, and therefore $\stem \models \tilde\alpha(\Psi) = 0$ according to
\autoref{lem-nullstellensatz-converse}.
By the premise,
\begin{align*}
q_\tau(V)^d \, s(V) \,\tilde\alpha(\Psi)(V') - \tilde\alpha(\Phi_\tau)(V)
\, \tilde\alpha(\Psi)(V) \in \left<\tau\right>
\end{align*}
for all $\tau \in \mathcal{T}$, therefore $q_\tau(V)^d s(V)
\tilde\alpha(\Psi)(V') \in \left< \tau, \tilde\alpha(\Psi)(V) \right>$
by \autoref{lem-ideal}, and from
\autoref{lem-nullstellensatz-converse} follows
\[
\tau(V, V') \land \tilde\alpha(\Psi)(V) = 0 \models q_\tau(V)^d \cdot s(V)
\cdot \tilde\alpha(\Psi)(V') = 0.
\]
Since $q_\tau$ is a common multiple of denominators of updates
in $\tau$ and $\neg \exit$ holds before any transition $\tau$, it
follows that $\neg \exit \models q_\tau(V) \neq 0$ by
\autoref{def-lasso}.  With $\neg \exit \models s(V) \neq 0$
we conclude that
\[
\tau(V, V') \land \tilde\alpha(\Psi)(V) = 0 \land \neg\exit
  \models \tilde\alpha(\Psi)(V') = 0. \quad\qed
\]
\end{proof}

A criterion for the method's completeness is given by the
following theorem.  The Nullstellensatz is applicable only when
one considers the theory of complex numbers, which in general
admits a proper subset of loop invariants.
Furthermore, the Nullstellensatz demands all involved ideals be
radical ideals~\cite{Cox91}.

\begin{theorem}[Completeness in $\complex$]\label{thm-completeness}
Let $L = (V, \stem, \mathcal{T}, \exit)$ be a polynomial lasso program
with the complex loop invariant\footnote{The assertions of
  \autoref{def-invariant} hold in the theory of the complex numbers.}
$p \in \reals[V]$.
If $\alpha: A \to \reals$ is a valuation such that
$\tilde\alpha(\Psi) = p$, then $\alpha$ can be extended to a solution
to the invariance condition if the following additional premises are
met.
\begin{enumerate}
\item\label{itm-completeness-l-det} The lasso $L$ is
  pseudo-deterministic.
\item\label{itm-completeness-radical} The ideal $\left<\stem\right>$ and the ideal $\langle p\rangle$
 are both radical ideals.
\item\label{itm-completeness-exit-principal} The ideal $\left<\exit\right>$ is
  generated by a single polynomial $s \in \reals[V]$.
\item\label{itm-completeness-extra} 
  The guard $h = 0$ of each transition $\tau\in\mathcal{T}$ is equivalent to $\mathit{True}$
  (i.e., $h$ is $0$).
\item\label{itm-completeness-ordering} The monomial ordering $\succ$
  is lexicographic and $x_i' \succ x_j$ for all $i$, $j$.
\end{enumerate}
\end{theorem}
\begin{proof}
The polynomial $p$ is a loop invariant of $L$, so by
\autoref{def-invariant},
\begin{align}
\stem &\models_\complex p(V) = 0 \text{ and} \label{eq-comp-inv-init} \\
p(V) = 0, \tau(V, V'), \neg\exit &\models_\complex p(V') = 0
\text{ for all } \tau \in \mathcal{T}. \label{eq-comp-inv-cons}
\end{align}
The ideal $\left<\stem\right>$ is a radical ideal by Premise~\ref{itm-completeness-radical}, so according to 
\hyperref[nullstellensatz]{Hilbert's Nullstellensatz},
Equation~\eqref{eq-comp-inv-init} implies
$p \in \left<\stem\right>$; and hence, $\alpha$ satisfies the first part of 
the invariance condition (IC).

To prove that $\alpha$ can be extended to satisfy the second part of IC,
note that Equation~\eqref{eq-comp-inv-cons}, combined with 
Premise~\ref{itm-completeness-exit-principal}, yields
\begin{eqnarray*}
 p(V) = 0, \tau(V, V') &\models_\complex & s(V) p(V') = 0.
\end{eqnarray*}
Using the Nullstellensatz, for some positive number $k$, we have
\begin{eqnarray*}
{\big( q_\tau(V)^d s(V) p(V') \big)^k} & \in & \langle p, \tau\rangle.
\end{eqnarray*}
Since $h$ is $0$, normalizing by $\tau$ is equivalent to 
replacing primed variables
using the update expressions in $\tau$, and hence,
\begin{eqnarray*}
 s(V)^k r(V)^k  \in  \langle p,\tau\rangle,
 \quad \mbox{where } r(V) := \gbdiv{q_\tau(V)^d p(V')}{\tau}
\end{eqnarray*}
Note that $r(V)$ has no prime variables since 
Premise~\ref{itm-completeness-ordering} ensures
all prime variables are greater with respect to the monomial ordering $\succ$ than the unprimed variables.
Therefore, $s(V)^k r(V)^k  \in  \langle p,\tau\rangle\cap \reals[V]$.
Now, there are two cases.
\\
{\em{(Case 1):
$\langle p,\tau\rangle\cap \reals[V] = \langle p\rangle$.}}
Then, it follows that $s(V)^k r(V)^k  \in  \langle p\rangle$.
Since $\langle p\rangle$ is a radical ideal, we can infer
 $s(V) r(V) \in \langle p\rangle$
 and hence
 $\gbdiv{q_\tau(V)^d s(V) p(V')}{\tau}$ is a multiple of $p$.
 Hence, second part of IC is satisfied.
\\
{\em{(Case 2):
$\langle p,\tau\rangle\cap \reals[V] \neq \langle p\rangle$.}}
This is possible only if some multiple of the 
denominators rewrites to $0$ by $p$.
Hence, $p = 0$ implies $s(V) = 0$ (since $s\neq 0$ implies
that denominators are nonzero).
Since $\langle p\rangle$ is a radical ideal, 
it follows $s\in \langle p\rangle$, and hence 
 $s(V) r(V) \in \langle p\rangle$ --- as in (Case 1) above.
\qed
\end{proof}

It is important to emphasize that the generic template polynomial for
the invariant must have a sufficiently large degree to be able to
specialize to the loop invariant.  This is presumed in the
completeness statement.  We will now discuss the other premises of
\autoref{thm-completeness}.

Premise~\ref{itm-completeness-l-det} ensures that 
the division of $\Psi(V')$
on a transition $\tau$ removes all primed variables, since we multiplied with $q_\tau(V)^d$ in the invariance condition.
Premise~\ref{itm-completeness-radical} is a requirement by
\hyperref[nullstellensatz]{Hilbert's Nullstellensatz}.  
In order to write a disjunction of $\exit$ and a polynomial equality
as a product, $\exit$ must have a single generator; this is stated in
Premise~\ref{itm-completeness-exit-principal}.  
We will discuss relaxing Premise~\ref{itm-completeness-extra} below.
Finally, Premise~\ref{itm-completeness-ordering} assures that primed variables are
eliminated first, leaving only unprimed variables in appropriate
cases.  This is relevant because the right hand side $\Phi_\tau(V)
\cdot \Psi(V)$ in the invariant condition contains only unprimed
variables.

\begin{remark}
We can generalize the completeness result to also include the case
when guards of transitions are nontrivial and 
when a conjunction
$p_1 = 0 \wedge p_2 = 0$ is an inductive invariant, but neither
$p_1 = 0$ nor $p_2 = 0$ by itself is an inductive invariant.
This requires generalizing the second part of the
invariance condition.
Let $\Psi_1$ and $\Psi_2$ be the templates whose instantiation 
gives $p_1$ and $p_2$ respectively. Then,
for all $\tau$ in $\mathcal{T}$, and for $i = 1,2$,
$$
\gbdiv{q_\tau(V)^d \cdot s(V) \cdot \Psi_i(V')}{\tau} = 
\Phi_1(V) \cdot \Psi_1(V) + \Phi_2(V)\cdot \Psi_2(V) +
\Phi_3(V) \cdot h_\tau(V) 
$$
Note that $\Phi_1,\Phi_2,\Phi_3$ are different templates
for different $\tau$'s and different $i$'s.
As before, the degrees of the templates are bounded by the
degree of the left-hand side, and $d$ is the total degree of $\Psi_i$.
In the completeness theorem, we can now drop
Premise~\ref{itm-completeness-extra}, but replace
Premise~\ref{itm-completeness-radical} by the following generalization:
\begin{itemize}
  \item[\ref{itm-completeness-radical}] 
    The ideal $\stem$, and for all $\tau$,
    the ideals $\langle p_1,p_2,h_\tau\rangle$,
    where $h_\tau=0$ is the guard of $\tau$, are radical ideals.
    Moreover, $\{p_1,p_2,h_\tau\}$ is a GB of 
    $\langle p_1,p_2,h_\tau\rangle$.
\end{itemize}
The proof of the new completeness claim is a natural
generalization of the proof of Theorem~\ref{thm-completeness} above.
\qed
\end{remark}

Besides the five restrictions of \autoref{thm-completeness},
completeness does not extend to the field of real numbers due to the
requirements of Hilbert's Nullstellensatz.  The underlying problem is
illustrated by the following example.

\begin{example}\label{ex-real-ideals}
The formula $\varphi \equiv x_1^2 + x_2^2 = 0$ has $x_1 = x_2 = 0$ as
its only solution over the reals.  However, $x_1, x_2 \notin \left<
x_1^2 + x_2^2 \right>$, although $\left< x_1^2 + x_2^2 \right>$ is a
radical ideal and $\varphi \models_\reals x_1 = 0, x_2 = 0$.
\end{example}
Alternatively,  we could formulate our results using
{\em{real radical ideals}}~\cite{Neuhaus98}.

Because the invariance condition in general is a non-linear
constraint, solving it might be very difficult. General approaches for
solving non-linear constraints have worst case space requirements that
are doubly exponential in the size of the input.
However, non-linear constraint solving is an active field of research
and recently there have been some promising efforts to take the
practical cases away from their {\small{DOUBLE-EXPSPACE}} worst-case complexity
bound~\cite{Jovanovic12}.

Another approach for solving the invariance condition stems from the
observation that the invarianc condition becomes linear if an assignment
for the template variables $(B_\tau)_{\tau \in \mathcal{T}}$ is given.
One could use heuristics to find this assignment.
For instance, practical experience suggests that if a solution to a variable $b \in
B_\tau$ is $\lambda_b \in \reals$, then the factor $(b -
\lambda_b)$ occurs somewhere in the invariance condition.  Using
factors in the former form as an initial guess for the variables
$(B_\tau)_{\tau \in \mathcal{T}}$ linearizes the equations and thus
enables quick discovery of a solution in some cases.

In the special case that $\Phi_\tau(V) := \lambda$ is degree $0$ 
(also called constant consecution),  $\lambda$ can be found as an
eigenvalue of an appropriate transformer 
constructed by interpreting bounded
degree polynomials as finite-dimensional vector spaces~\cite{Rebiha}.


\section{Synthesis}
\label{sec-synthesis}

The technique for finding a loop invariant using the invariance
condition established in the previous section will now be used for
program synthesis.  Given a polynomial lasso program, some transition
updates can be parameterized by replacing them with template
polynomials.
The synthesis process will try to find a valuation of these template
variables while respecting some post condition.  The following
definition formalizes this concept.

\begin{definition}[Synthesis Problem]\label{def-synthesis-problem}
A \emph{synthesis problem} $S = (C, L, \post)$ consists of
\begin{itemize}
\item a set of \emph{synthesis variables} $C$,
\item a polynomial lasso program $L = (V, \stem, \mathcal{T}, \exit)$ where $\stem$ and $\tau \in \mathcal{T}$ contain template polynomials over $(C, V)$, and
\item a \emph{post condition} in form of an algebraic assertion $\post$ over $V$.
\end{itemize}

A \emph{solution to the synthesis problem $S$} is a valuation $\alpha:
C \to \reals$ such that the lasso $L_\alpha = (V,
\tilde\alpha(\stem), \tilde\alpha(\mathcal{T}), \exit)$ is partially
correct with respect to the post condition $\post$.
\end{definition}

\begin{example}\label{ex-product4}
Transforming $L$ from \autoref{ex-product1} to $L'$ by changing
the transition $\tau$ to
\[
y' = y - 1 \land s' = c_1 x_0 + c_2 y_0 + c_3 y + c_4 s + c_5
\]
gives rise to a synthesis problem $S = (C, L', \post)$ for $C = \{
c_1, c_2, c_3, c_4, c_5 \}$ and $\post \equiv s = x_0 y_0$.  A
solution to $S$ is $\alpha: c_1 \mapsto 1, c_2 \mapsto 0, c_3 \mapsto
0, c_4 \mapsto 1, c_5 \mapsto 0$ since $L_\alpha = L$ and $L$ is
partially correct with respect to $\post$ according to
\autoref{ex-product2}.
\end{example}

Our approach for solving the synthesis problem is based on the
technique from the previous section.  We will prove the partial
correctness of the synthesized lasso program.  The following lemma states that
synthesized polynomial lasso program will be partially correct.

\begin{lemma}[Synthesis Solution]\label{lem-synthesis-solution}
Let $S = (C, L, \post)$ be a synthesis problem, $\alpha: C \to
\reals$ be a valuation on the synthesis variables, and let $p$ be
an invariant for $L_\alpha$.  If $p = 0 \land \exit \models \post$, then
$L_\alpha$ is partially correct with respect to $\post$, i.e., $\alpha$ is a
solution to $S$.
\end{lemma}
\begin{proof}
Let $\sigma = \nu_0 \ldots \nu_k$ be a finite execution of $L_\alpha$.
According to the assumption, $p$ is an invariant of $L_\alpha$, so by
\autoref{def-invariant}, $\nu_i \models p = 0$ for all $0 \leq i \leq
k$.  By \autoref{def-lasso-semantics}, $\nu_k \models \exit$, therefore
$\nu_k \models p = 0 \land \exit$. According to the assumption, this
implies $\nu_k \models \post$, which proves the correctness of
$L_\alpha$. \qed
\end{proof}

To find a valuation for the synthesis variables, we define a
\emph{synthesis condition}.  The synthesis condition will constrain
the synthesis variables so that existence of a loop invariant $p$ that
implies the post condition is guaranteed; that is,
\begin{align}\label{eq-synthesis-post}
p = 0 \land \exit \models \post.
\end{align}
If $\post = \bigwedge_i \post_i = 0$, then the above is implied by
$\post_i \in \left< p, \exit \right>$ by
\autoref{lem-nullstellensatz-converse}.  However, computing the
Gröbner basis with respect to a template polynomial for $p$ is
extremely inefficient and potentially involves a huge number of case
splits.
But according to \autoref{lem-ideal}, we can equivalently write
\begin{align}\label{eq-synthesis-post2}
\post_i - t p \in \exit,
\end{align}
for some unknown $t \in \reals[V]$.  This enables us to rewrite
\eqref{eq-synthesis-post} in a way that only involves computing the
Gröbner basis for non-template polynomials.

\begin{example}\label{ex-product5}
Let $S$ be the synthesis problem from \autoref{ex-product4}.  We use
the loop invariant $p = s + x_0 y - x_0 y_0$ from
\autoref{ex-product2} in \autoref{lem-synthesis-solution} to show that
$\alpha$ is a solution to $S$ by checking
\[
s + x_0 y - x_0 y_0 = 0 \land y = 0 \models s = x_0 y_0,
\]
or instead that for $t = 1$,
\[
(s - x_0 y_0) - t (s + x_0 y - x_0 y_0) \in \exit.
\]
\end{example}

\begin{definition}[Synthesis Condition]\label{def-synthesis-condition}
Let $S = (C, L, \bigwedge_{i=1}^m \post_i(V) \!\!=\!\! 0)$ be a synthesis
problem, let $\Psi$ be a template polynomial over $(A, V)$ and for all
$0 \leq i \leq m$, let $\Omega_i$ be a template polynomial over
$(D_i, V)$.  The \emph{synthesis condition},
$\syc(S, \Psi, \{\Omega_i \mid 0 \leq i \leq m \})$,
of $S$ is the formula
\begin{align*}
\inc(L, \Psi) \land \bigwedge_i \gbdiv{\post_i(V) - \Omega_i(V)
  \Psi(V)}{\exit} = 0
\end{align*}
\end{definition}

Following the same argument as for the invariant condition,
the synthesis condition simplifies to a conjunction 
of non-linear equations in the
variables $A \cup C \cup (\bigcup_{\tau \in \mathcal{T}} B_\tau) \cup
(\bigcup_{i=1}^m D_i)$.  Utilizing an SMT solver, a solution to this
constraint can be obtained that is then used to instantiate the template
polynomials in the loop invariant and polynomial lasso program.
According to the next theorem, this yields a correct program instance.

Motivated by \autoref{ex-product5}, we may set $\Omega_i = 1$ in the
synthesis condition.  In this case the constraint $\gbdiv{\post_i(V) -
  \Psi(V)}{\exit} = 0$, the synthesis condition's constraint
corresponding to the post condition, is linear.  Using this
observation we can use linear methods to eliminate some variables from
the constraint system, thus simplifying it.  The same trick also
applies to the constraint $\gbdiv{\Psi(V)}{\stem} = 0$ in the
invariance condition if the program stem does not contain
any synthesis variables $C$.

Because the coefficients of some of the polynomials in $L$ contain
template variables, special care must be taken when computing a
Gröbner basis for $\stem$ or $\tau \in \mathcal{T}$.  Every division
by some term containing a variable demands a case split on whether
this term evaluates to zero.  One way of circumventing this problem is
to compute a Gröbner basis where the underlying algebraic structure for polynomial coefficients is the ring of parameter polynomials $\reals[A]$.
This requires a slightly modified division algorithm~\cite{Adams94,BachmairTiwari97:RTAsmall}.

\begin{theorem}[Synthesis Soundness]\label{def-synthesis-soundness}
If $\alpha: A \cup (\bigcup_{\tau \in \mathcal{T}} B_\tau) \cup C \cup
(\bigcup_i D_i) \rightarrow \reals$ is an assignment for the
template variables that models the synthesis condition, then
$L_\alpha$ is partially correct with respect to the post condition
$\post$ and $\tilde\alpha(\Psi)$ is an invariant of $L_\alpha$.
\end{theorem}
\begin{proof}
By \autoref{thm-soundness}, $\tilde\alpha(\Psi)$ is an invariant of
$L_\alpha$.  By definition, $\post_i(V) - \tilde\alpha(\Omega_i)(V)
\tilde\alpha(\Psi)(V) \in \left<\exit\right>$, therefore $\post_i \in \left<
\tilde\alpha(\Psi), \exit \right>$ for all $i$ according to
\autoref{lem-ideal}. By \autoref{lem-nullstellensatz-converse},
$\tilde\alpha(\Psi) \land \exit \models \post_i$ for all $i$, hence
$\tilde\alpha(\Psi) \land \exit \models \post$.
\autoref{lem-synthesis-solution} ensures that this implies that
$L_\alpha$ is partially correct. \qed
\end{proof}

The synthesis process is not complete, even when the restrictions of
\autoref{thm-completeness} hold.  The reason for this is the
polynomial $t$ in \eqref{eq-synthesis-post2}: we are using templates
$\Omega_i$ for $t$, but a priori we have no upper bound on the degree of 
$t$.
In practice, a template of degree $0$ might be sufficient, as in our examples (see \autoref{sec-benchmark}).


\section{Termination}
\label{sec-termination}

A solution to the synthesis problem guarantees partial correctness of the
synthesized program; however, termination is not guaranteed. Even if
the synthesized program terminates, it might be highly inefficient,
going through unnecessarily many loop iterations.

\begin{example}\label{ex-product6}
If one extends $L'$ from \autoref{ex-product4} to $L''$ by
changing $\tau$ to
\[
y' = c_6 y + c_7 \land s' = c_1 x_0 + c_2 y_0 + c_3 y + c_4 s + c_5
\]
this yields a synthesis problem $S' = (C', L'', \post)$ with $C' = C
\cup \{ c_6, c_7 \}$. Possible solutions to $S'$ include the
valuations $\alpha_\lambda: c_1 \mapsto \lambda, c_2 \mapsto 0, c_3
\mapsto 0, c_4 \mapsto 1, c_5 \mapsto 0, c_6 \mapsto 1, c_7 \mapsto
-\lambda$ for all $\lambda \in \reals$.

If $\lambda$ is small, the program needs more iterations for the same
input, and if $\lambda$ is zero, $L_\alpha$ will not terminate at all.
\end{example}

In order to address this, the synthesis condition can be augmented
with a series of test cases, predefined input-output pairs that
explicitly state the transitions required to compute them.

\begin{definition}[Test Case]\label{def-testcase}
Let $(C, L, \post)$ be a synthesis problem where
$L = (V, \stem, \mathcal{T}, \exit)$ is a pseudo-deterministic
polynomial lasso program containing template variables $C$.
A \emph{test case}
$t = (\nu_0, \nu, \tau_1 \ldots \tau_k)$ consists of two
$V$-valuations $\nu_0$ and $\nu$ corresponding to the initial and
final state respectively such that $\nu \models \exit$, as well as a
finite sequence of transitions $\tau_1, \ldots, \tau_k \in
\mathcal{T}$.
A solution $\alpha$ to a synthesis problem $S$ is said to \emph{adhere
  to the test case $t$} if, for $\nu_i = \tau_i \circ \ldots \circ
\tau_1(\nu_0)$, the sequence $\sigma = \nu_0 \nu_1 \ldots \nu_k$ is an
execution of $L_\alpha$ and $\nu_k = \nu$.
\end{definition}

\begin{lemma}\label{lem-testcase}
Let $S = (C, L, \post)$ be a synthesis problem with solution $\alpha$
and let $t = (\nu_0, \nu, \tau_1 \ldots \tau_k)$ be a test case.  If
\begin{align}
\nu_0 &\models \alpha(\stem), \label{eq-testcase1} \\
\nu &= \alpha(\tau_k) \circ \ldots \circ \alpha(\tau_1)(\nu_0),
\label{eq-testcase2} \\
\nu_i &\not\models \exit \text{ for } 0 \leq i \leq k - 1, \text{ and}
\label{eq-testcase3} \\
\nu_k &\models \exit \label{eq-testcase4}
\end{align}
then $L_\alpha$ adheres to the test case $t$.
\end{lemma}
\begin{proof}
$\sigma = \nu_0 \nu_1 \ldots \nu_k$ for $\nu_i = \alpha(\tau_i) \circ
  \ldots \circ \alpha(\tau_1)(\nu_0)$ is by construction an execution
of $L_\alpha$ according to \autoref{def-lasso-semantics}.  From
\eqref{eq-testcase2} follows that $\nu_k = \nu$. \qed
\end{proof}

If we add the equations \eqref{eq-testcase1}, \eqref{eq-testcase2},
\eqref{eq-testcase3} and \eqref{eq-testcase4} to the synthesis
condition for every given test case, then by \autoref{lem-testcase}
any solution to these constraints will yield a solution to $S$ that
adheres to the test cases.

\begin{example}\label{ex-product7}
Consider the synthesis problem $S'$ from \autoref{ex-product6}.  The
execution $\sigma$ from \autoref{ex-product1} gives rise to the test
case $t = (\nu_0, \nu_1, \tau)$, which by \autoref{lem-testcase} adds
the following additional constraints on the synthesis condition.
\begin{align*}
\begin{aligned}
1 &= 1 \\
0 &= 0
\end{aligned} &&
\begin{aligned}
0 &= 1c_6 + c_7 \\
3 &= 3c_1 + 1c_2 + 1c_3 + 0c_4 + c_5
\end{aligned} &&
\begin{aligned}
1 &\neq 0 \\
0 &= 0
\end{aligned}
\end{align*}
The valuation $\alpha_1$ is the only one of the valuations $\alpha_\lambda$ given in \autoref{ex-product6} that models these two equations (however, it is not the only possible solution).
$L_{\alpha_1}$ is a terminating lasso program for positive
integers $y_0$.
\end{example}

In theory, if it is possible to synthesize a terminating program,
then there exists a finite set of test cases that will guarantee
that a terminating lasso is synthesized.
\begin{theorem}\label{thm-finite-testcases}
Let $S = (C, L, \post)$ be a synthesis problem.  If there is a solution
$\alpha$ to $S$ such that $L_\alpha$ is terminating then there is a
finite set of test cases $\Sigma$ such that any solution $\beta$ of
$S$ which adheres to all test cases $t \in \Sigma$ is terminating.
\end{theorem}
\begin{proof}
Let $\Sigma = \{ t_0, t_1, \ldots \}$ be the test cases to all
possible executions of $L_\alpha$, and assume $\Sigma$ is infinite
(otherwise there is nothing to show).
Each test case $t \in \Sigma$ corresponds to a polynomial assertion
over the variables $C$ by \eqref{eq-testcase1} and
\eqref{eq-testcase2}.  This assertion constrains
possible assignments of $C$.  For every $i \geq 0$, let $\Sigma_i
= \{ t_0, \ldots, t_i \} \subset \Sigma$ be an ascending chain of
finite subsets of $\Sigma$ and let $I_i$ be the ideal generated by the
assertions from the test cases of $\Sigma_i$.
It is clear that $\Sigma_i \subset \Sigma_{i+1}$, and hence $I_i
\subseteq I_{i+1}$.  By the \emph{Ascending Chain
  Condition}~\cite{Cox91}, the ascending chain of ideals $I_0
\subseteq I_1 \subseteq \ldots$ must become stationary for some
integer $k$, meaning $I_k = I_i$ for all $i \geq k$.  This implies
that the finite set of test cases $\Sigma_k$ corresponds to the same
ideal as $\Sigma_i$ for $i \geq k$ and hence they have the same
solution (set of assignments) for $C$.  
As a consequence, any solution $\beta$
to $S$ adhering to the test cases from $\Sigma_k$ will enforce that
$\sigma$ is an execution of $L_\beta$ iff it is an execution of
$L_\alpha$. \qed
\end{proof}

While \autoref{thm-finite-testcases} assures that under any
circumstances, a finite set of test cases $\Sigma$ suffices to force a useful
solution from the synthesis problem, no upper bound to the cardinality
of $\Sigma$ is given.

In theory, this provides us with two powerful approaches of generating
polynomial lasso programs, given an a priori bound on the number of program
variables $V$. Both involve creating a polynomial lasso program with generic
template polynomials as updates and guards.
\begin{enumerate}
\item Specify a (large) number of test cases.  Ideally, these test
  cases can be automatically generated in some sophisticated way that
  ensures that they are not too redundant.
\item Provide a post condition and a complexity guess.  Using the
  complexity guess, a terminating skeleton of the synthesis problem is
  generated using counter variables.  The post condition provides a
  statement regarding the program's purpose.
\end{enumerate}
Needless to say, both approaches create very large
synthesis conditions that are unlikely to be handled automatically
by present-day non-linear solvers, but this can 
change, especially for small program fragments, as technology develops.


\section{Experimental Evaluation}
\label{sec-benchmark}

We implemented our method in Haskell and used \texttt{nlsat}~\cite{Jovanovic12}\footnote{As implemented in z3 version 4.3.1.~\url{http://z3.codeplex.com/}} to solve the non-linear constraints.
To evaluate the practicability and scalability of our method, we ran
it on a few selected examples which are listed in
\autoref{table-examples} together with a short description.
Each example translates to a pseudo-deterministic polynomial lasso program.
\ifappendix
See the \hyperref[sec-appendix-code]{appendix} for the source code to the examples
as well as the discovered solutions to the constraints.
\fi

\begin{table}[t]
\centering
\tabcolsep=2mm
\begin{tabular}{l p{98mm}}
\emph{name} & \emph{description}
\\ \hline
\texttt{product} & multiplication of two integers by repeated addition (see \autoref{fig-product} and \autoref{ex-product1}) \\
\texttt{productS} & \texttt{product} with synthesis of one update
statement (see \autoref{ex-product4}) \\
\texttt{productSY} & \texttt{product} with synthesis of the loop body,
including the termination-critical variable $y$
(see \autoref{ex-product6}) \\
\texttt{product2} & \texttt{product} with reciprocal $y$ \\
\texttt{product2S} & \texttt{product2} with synthesis of one update
statement \\
\texttt{gcd\_lcm}  & greatest common denominator and least common multiple
  of two integers~\cite{Sankaranarayanan04} \\
\texttt{gcd\_lcmS} & \texttt{gcd\_lcm} with synthesis of two update
statements \\
\texttt{div\_mod}  & integer division with remainder~\cite{Dijkstra76} \\
\texttt{div\_modS} & \texttt{div\_mod} with synthesis of the complete
loop body with linear updates \\
\texttt{root2} & integer square root~\cite{Dijkstra76} \\
\texttt{root2S} & \texttt{root2} with synthesis of the stem and one update statement \\
\texttt{squareS} & square of an integer synthesized from a terminating
skeleton with linear assignments \\
\texttt{cubeS} & cube of an integer synthesized from a terminating
skeleton with linear assignments
\end{tabular}
\caption{Example programs used to perform verification and synthesis
experiments described in \autoref{table-benchmark}.
\ifappendix
For the source code to the examples see the \hyperref[sec-appendix-code]{appendix}.
\fi
}
\label{table-examples}
\end{table}

\autoref{table-benchmark} contains the experiment's results.  We list
the program name together with the number of synthesis variables
($\#C$), the degree of the loop invariant's generic template
polynomial ($\deg$), the number of its template variables ($\#A$), the
total number of variables in the generated constraint ($\#$vars), the
number of test cases used ($\#$tc), the time to generate the constraint in seconds (constraints time) and the running time of the SMT solver in seconds (solver time).
Our test system was a computer with eight AMD Opteron 8220 2.80GHz
CPUs and 32GB RAM.

\begin{table}[t]
\centering
\tabcolsep=1mm
\begin{tabular}{lrrrrrrr}
\emph{name} & $\#C$ & $\deg$ & $\#A$ & $\#$vars & $\#$tc
  & \emph{constraints time (s)} & \emph{solver time (s)} \\
\hline
\texttt{product}   &  0 & 2 & 15 & 20 & 0 &   0.55 &   0.02 \\
\texttt{productS}  &  5 & 2 & 15 & 25 & 0 &   1.47 &   0.01 \\
\texttt{productSY} &  7 & 2 & 15 & 27 & 2 &   3.39 &   0.02 \\
\texttt{product2}  &  0 & 3 & 35 & 50 & 0 &  39.23 & 128.24 \\
\texttt{product2S} &  5 & 3 & 35 & 55 & 0 & 200.20 &  24.46 \\
\texttt{gcd\_lcm}  &  0 & 2 & 28 & 42 & 0 &  11.85 &   0.02 \\
\texttt{gcd\_lcmS} & 10 & 2 & 28 & 52 & 0 &  17.01 &   0.01 \\
\texttt{div\_mod}  &  0 & 2 & 15 & 16 & 0 &   0.62 &   0.01 \\
\texttt{div\_modS} & 10 & 2 & 15 & 26 & 5 &  10.03 &   0.03 \\
\texttt{root2}     &  0 & 2 & 15 & 25 & 0 &   2.80 &   4.52 \\
\texttt{root2S}    &  9 & 2 & 15 & 34 & 0 &   3.80 &   0.02 \\
\texttt{squareS}   &  6 & 2 & 10 & 20 & 0 &   0.56 &   0.00 \\
\texttt{cubeS}     & 14 & 3 & 35 & 54 & 0 &  90.88 &  41.05 \\
\end{tabular}
\caption{Experimental results showing the time required to
  verify/synthesize various example programs, along with the size of
  the non-linear constraints solved in the
  process.}\label{table-benchmark}
\end{table}

While the synthesis process is very fast for small examples, the
non-linear constraint solver becomes the bottleneck in
medium-sized problems (\texttt{product2}, \texttt{product2S} and
\texttt{cubeS} use generic templates of degree $3$):
solving non-linear constraints scales poorly with the number of variables involved.
Test cases might help mitigate this issue by significantly reducing the solution space.


\section{Conclusion}
\label{sec-conclusion}

We presented a method for synthesizing polynomial programs.  This
method is based on the discovery of non-linear loop invariants that
prove the program's correctness.  We generate a \emph{synthesis
  condition}, a non-linear constraint whose solution is the
synthesized polynomial lasso program and a loop invariant.  We extended
existing methods for non-linear invariant generation and provided a
completeness criterion (\autoref{thm-completeness}).
If we synthesize update statements of variables
that occur in the exit condition, termination becomes a concern.
We showed that we can utilize a finite
set of test cases to restrict the solution space to terminating lassos
(\autoref{thm-finite-testcases}).

Using a benchmark of small examples, we showed that our method is
applicable for the synthesis of small programs, as well as parts of
medium-sized ones.  A resource bottleneck is the non-linear
constraint solver.  As the solving of non-linear constraints is an
active area of research, we expect that our technique will become more
effective as non-linear solvers improve.

We assumed that the programs' variables take values in the set of reals
$\reals$, but since Gröbner bases are computable over
rings~\cite{Adams94,BachmairTiwari97:RTAsmall}, our method can also be
applied to the integers $\integers$ or the finite ring of machine
integers $\integers_{2^w}$ (see also \cite{Seidl08}).
Future work could also consider the question of how this method can be
improved to handle inequalities.

\subsubsection*{Acknowledgements}
We would like to thank the anonymous reviewers for their valuable
feedback.


\bibliographystyle{abbrv}
\bibliography{references}


\ifappendix
\newpage
\section*{Appendix: Source Code to the Experiments}
\label{sec-appendix-code}

This appendix lists the source code for the example programs used in
\autoref{sec-benchmark} (compare \autoref{table-examples} for a short
description).  The programs are given in pseudo-code rather than
polynomial lasso programs for improved readability and for completeness: some
parts of the program code have to be omitted in the
translation to polynomial lasso programs
(e.g.\ the exit condition in \texttt{div\_mod} and
\texttt{div\_modS} is an inequality).
However, this translation is straightforward.
The result is a \emph{pseudo-deterministic} polynomial lasso program in each example.
We provide the assignment $\alpha$ to the generated constraints as found by the SMT solver.
In our examples, we set $\Omega_i = 1$ in the synthesis condition.

\dline \\
For the source code to the program \texttt{product} see
\autoref{fig-product}.  Also compare \autoref{ex-product1} and
\autoref{ex-product3}.  We use the generic template of degree $2$ over
the variables $V = \{ x_0, y_0, y, s \}$.
\begin{align*}
\tilde\alpha(\Psi) &= \frac{1}{4} (x_0 y_0 - x_0 y - s) \\
\tilde\alpha(\Phi) &= y
\end{align*}

\dline
\begin{lstlisting}
procedure productS($x_0$, $y_0$):
    $s$ := $0$;
    $y$ := $y_0$;
    while ($y \neq 0$):
        $s$ := $c_0 x_0 + c_1 y_0 + c_2 y + c_3 s + c_4$;
        $y$ := $y - 1$;
    assert($s = x_0 y_0$);
    return $s$;
\end{lstlisting}
From \autoref{ex-product4}.  $C = \{ c_0, c_1, c_2, c_3, c_4 \}$; we
use the generic template of degree $2$ over the variables $V = \{ x_0,
y_0, y, s \}$.
\begin{align*}
\alpha(c_0) &= 1, \quad \alpha(c_1) = \alpha(c_2) = 0, \quad
\alpha(c_3) = 1, \quad \alpha(c_4) = 0 \\
\tilde\alpha(\Psi) &= s - x_0 y_0 + x_0 y \\
\tilde\alpha(\Phi) &= y
\end{align*}

\dline
\begin{lstlisting}
procedure productSY($x_0$, $y_0$):
    $s$ := $0$;
    $y$ := $y_0$;
    while ($y \neq 0$):
        $s$ := $c_0 x_0 + c_1 y_0 + c_2 y + c_3 s + c_4$;
        $y$ := $c_5 y + c_6$;
    assert($s = x_0 y_0$);
    return $s$;
\end{lstlisting}
Test cases:
\begin{itemize}
\item \texttt{productSY(3, 1) == 1} (1 loop iteration)
\item \texttt{productSY(3, 2) == 6} (2 loop iterations)
\end{itemize}
From \autoref{ex-product6}.  $C = \{ c_0, c_1, c_2, c_3, c_4 \}$; we
use generic template of degree $2$ over the variables $V = \{ x_0,
y_0, y, s \}$.
\begin{align*}
\alpha(c_0) &= 1, \quad \alpha(c_1) = -\frac{1}{2}, \quad
\alpha(c_2) = 1, \quad \alpha(c_3) = 1, \quad
\alpha(c_4) = -\frac{1}{2}, \\
\alpha(c_5) &= 1, \quad \alpha(c_6) = -1 \\
\tilde\alpha(\Psi) &= s - x_0 y_0 + x_0 y + \frac{1}{2} y^2
  - \frac{1}{2} y_0 y \\
\tilde\alpha(\Phi) &= y
\end{align*}

\dline
\begin{lstlisting}
procedure product2($x_0$, $y_0$):
    $s$ := $x_0$;
    $y$ := $\frac{1}{y_0}$
    while ($y \neq 1$):
        $s$ := $s + x_0$;
        $y$ := $\frac{y}{1 - y}$;
    return $s$;
\end{lstlisting}
This example differs from \texttt{product} by the use of the
reciprocal value of $y$.  The assignments \texttt{$y$ :=
  $\frac{1}{y_0}$;} and \texttt{$y$ := $\frac{y}{1 - y}$;} are
translated to the polynomials $y' y_0 - 1$ and $y' (1 - y) - y$
respectively.  We use the generic template of degree $3$ over the
variables $V = \{ x_0, y_0, y, s \}$.
\begin{align*}
\tilde\alpha(\Psi) &= \sqrt{2} (x_0 y - x_0 + x_0 y_0 y
- y s) \\
\tilde\alpha(\Phi) &= y - 1
\end{align*}

\dline
\begin{lstlisting}
procedure product2S($x_0$, $y_0$):
    $s$ := $x_0$;
    $y$ := $\frac{1}{y_0}$
    while ($y \neq 1$):
        $s$ := $c_0 x_0 + c_1 y_0 + c_2 y + c_3 s + c_4$;
        $y$ := $\frac{y}{1 - y}$;
    assert($s = x_0 y_0$);
    return $s$;
\end{lstlisting}
We use the generic template of degree $3$ over the variables $V = \{
x_0, y_0, y, s \}$.
\begin{align*}
\alpha(c_0) &= 1, \quad \alpha(c_1) = \alpha(c_2) = 0, \quad
\alpha(c_3) = 1, \quad \alpha(c_4) = 0 \\
\tilde\alpha(\Psi) &= x_0 - x_0 y - x_0 y_0 y + y s \\
\tilde\alpha(\Phi) &= y - 1
\end{align*}

\dline
\begin{lstlisting}
procedure gcd_lcm($x_1$, $x_2$):
    $y_1$ := $x_1$;
    $y_2$ := $x_2$;
    $y_3$ := $x_2$;
    $y_4$ := $0$;
    while ($y_1 \neq y_2$):
        if ($y_1 > y_2$):
            $y_1$ := $y_1 - y_2$;
            $y_4$ := $y_4 + y_3$;
        else:
            $y_2$ := $y_2 - y_1$;
            $y_3$ := $y_3 + y_4$;
    assert($y_1 (y_3 + y_4) - x_1 x_2$);
    return $(y_1, y_3 + y_4)$;
\end{lstlisting}
The inequality $y_1 > y_2$ cannot be translated into a polynomial lasso program
syntax and is thus omitted.
We use the generic template of degree $2$ over the
variables $V = \{ x_1, x_2, y_1, y_2, y_3, y_4 \}$.
\begin{align*}
\tilde\alpha(\Psi) &= x_1 x_2 - y_1 y_3 - y_2 y_4 \\
\tilde\alpha(\Phi_1) = \tilde\alpha(\Phi_2) &= y_1 - y_2 
\end{align*}

\dline
\begin{lstlisting}
procedure gcd_lcmS($x_1$, $x_2$):
    $y_1$ := $x_1$;
    $y_2$ := $x_2$;
    $y_3$ := $x_2$;
    $y_4$ := $0$;
    while ($y_1 \neq y_2$):
        if ($y_1 > y_2$):
            $y_4$ := $c_0 y_1 + c_1 y_2 + c_2 y_3 + c_3 y_4 + c_4$;
            $y_1$ := $y_1 - y_2$;
        else:
            $y_3$ := $c_5 y_1 + c_6 y_2 + c_7 y_3 + c_8 y_4 + c_9$;
            $y_2$ := $y_2 - y_1$;
    assert($y_1 (y_3 + y_4) - x_1 x_2$);
    return $(y_1, y_3 + y_4)$;
\end{lstlisting}
The inequality $y_1 > y_2$ cannot be translated into a polynomial lasso program
syntax and is thus omitted.  $C = \{ c_0, c_1, c_2,
c_3, c_4, c_5, c_6, c_7, c_8, c_9 \}$; we use the generic template of
degree $2$ over the variables $V = \{ x_1, x_2, y_1, y_2, y_3, y_4
\}$.
\begin{align*}
\alpha(c_0) = \alpha(c_1) &= 0, \quad \alpha(c_2) = \alpha(c_3) = 1,
\quad \alpha(c_4) = 0, \\
\alpha(c_5) = \alpha(c_6) &= 0, \quad \alpha(c_7) = \alpha(c_8) = 1,
\quad \alpha(c_9) = 0, \\
\tilde\alpha(\Psi) &= x_1 x_2 - y_1 y_3 - y_2 y_4 \\
\tilde\alpha(\Phi_1) = \tilde\alpha(\Phi_2) &= y_1 - y_2 
\end{align*}

\dline
\begin{lstlisting}
procedure div_mod($a$, $d$):
    $q$ := $0$;
    $r$ := $a$;
    while ($r \geq d$):
        $r$ := $r - d$;
        $q$ := $q + 1$;
    return $(q, r)$;
\end{lstlisting}
The while-condition $r \geq d$ is translated to $true$.
We use the generic template of degree $2$ over the variables $V = \{
a, d, q, r \}$.
\begin{align*}
\tilde\alpha(\Psi) &= \frac{1}{2} (r + qd - a) \\
\tilde\alpha(\Phi) &= 1
\end{align*}

\dline
\begin{lstlisting}
procedure div_modS($a$, $d$):
    $q$ := $0$;
    $r$ := $a$;
    while ($r \geq d$):
        $r$ := $c_0 a + c_1 d + c_2 q + c_3 r + c_4$;
        $q$ := $c_5 a + c_6 d + c_7 q + c_8 r_{old} + c_9$;
    return $(q, r)$;
\end{lstlisting}
Test cases:
\begin{itemize}
\item \texttt{div\_modS(4, 3) == (1, 1)} (1 loop iteration)
\item \texttt{div\_modS(5, 2) == (2, 1)} (2 loop iterations)
\item \texttt{div\_modS(1, 1) == (1, 0)} (1 loop iteration)
\item \texttt{div\_modS(15, 6) == (2, 3)} (2 loop iterations)
\item \texttt{div\_modS(17, 17) == (1, 0)} (1 loop iterations)
\end{itemize}
The while-condition $r \geq d$ is translated to $true$.  The synthesis
relies purely on test cases for correctness and termination as there
are no exit or post condition supplied.  $C = \{ c_0, c_1, c_2, c_3,
c_4, c_5, c_6, c_7, c_8, c_9 \}$; we use the generic template of
degree $2$ over the variables $V = \{ a, d, q, r \}$.
\begin{align*}
\alpha(c_0) &= 0, \quad \alpha(c_1) = -1, \quad \alpha(c_2) = 0, \quad
\alpha(c_3) = 1, \quad \alpha(c_4) = 0, \\
\alpha(c_5) &= 0, \quad \alpha(c_6) = 0, \quad \alpha(c_7) = 1, \quad
\alpha(c_8) = 0, \quad \alpha(c_9) = 1 \\
\tilde\alpha(\Psi) &= r + dq - a \\
\tilde\alpha(\Phi) &= 1
\end{align*}

\dline
\begin{lstlisting}
procedure root2($n$):
    $p$ := $0$;
    $q$ := $1$;
    $r$ := $n$;
    while ($q \leq n$):
        $q$ := $4q$;
    while ($q \neq 1$):
        $q$ := $\frac{q}{4}$;
        $h$ := $p + q$;
        $p$ := $\frac{p}{2}$;
        if ($r \geq h$):
            $p$ := $p + q$;
            $r$ := $r - h$;
    assert($n = p^2 + r$);
    return $(p, r)$;
\end{lstlisting}
The first while loop is translated to an arbitrary assignment to $q$,
the \texttt{if}-condition is omitted.  We use the generic template of degree
$2$ over the variables $V = \{ p, q, r, n \}$.
\begin{align*}
\tilde\alpha(\Psi) &= \frac{1}{8} (nq - qr - p^2)\\
\tilde\alpha(\Phi_1) =
\tilde\alpha(\Phi_2) &= \frac{1}{4} (q - 1)
\end{align*}

\dline
\begin{lstlisting}
procedure root2S($n$):
    $p$ := $c_0 n + c_1$;
    $q$ := $1$;
    $r$ := $c_2 n + c_3$;
    while ($q \leq n$):
        $q$ := $4q$;
    while ($q \neq 1$):
        $q$ := $\frac{q}{4}$;
        $h$ := $p + q$;
        $p$ := $\frac{p}{2}$;
        if ($r \geq h$):
            $p$ := $p + q$;
            $r$ := $c_4 r + c_5 p_{old} + c_6 q_{old} + c_7 n + c_8$;
    assert($n = p^2 + r$);
    return $(p, r)$;
\end{lstlisting}
$C = \{ c_0, c_1, c_2, c_3, c_4, c_5, c_6, c_7, c_8 \}$.  The first
\texttt{while} loop is translated to an arbitrary assignment to $q$,
the \texttt{if}-condition is omitted.
This example has a parameterized program stem.
We use the generic template of degree $2$ over the
variables $V = \{ p, q, r, n \}$.
\begin{align*}
\alpha(c_0) = \alpha(c_1) &= 0, \quad \alpha(c_2) = 1, \alpha(c_3) =
0, \\
\alpha(c_4) &= 1, \quad \alpha(c_5) = -1, \quad \alpha(c_6) =
\frac{1}{4}, \quad
\alpha(c_7) = \alpha(c_8) = 0 \\
\tilde\alpha(\Psi) &= nq - qr - p^2 \\
\tilde\alpha(\Phi_1) = \tilde\alpha(\Phi_2) &= \frac{1}{4} (q - 1)
\end{align*}

\dline
\begin{lstlisting}
procedure squareS($n$):
    $a$ := $n$;
    $b$ := $c_0 n + c_1$;
    while ($a \neq 0$):
        $b$ := $c_2 a + c_3 b + c_4 n + c_5$;
        $a$ := $a - 1$;
    assert($b = n^2$);
    return $b$;
\end{lstlisting}
$C = \{ c_0, c_1, c_2, c_3, c_4, c_5 \}$.  This is an example for a
synthesis from terminating lasso program skeleton.
The lasso is constructed from
a complexity guess ($n$ steps to termination) and an auxiliary
variable $b$ (see also the program \texttt{cubeS}).  We use the
generic template of degree $2$ over the variables $V = \{ a, b, n \}$.
\begin{align*}
\alpha(c_0) &= -\frac{1}{2}, \quad \alpha(c_1) = 0, \\
\alpha(c_2) &= 2, \quad \alpha(c_3) = 1, \quad \alpha(c_4) = 0, \quad
\alpha(c_5) = -\frac{1}{2} \\
\tilde\alpha(\Psi) &= b - n^2 + \frac{1}{2} a + a^2 \\
\tilde\alpha(\Phi) &= a
\end{align*}

\dline
\begin{lstlisting}
procedure cubeS($n$):
    $a$ := $n$;
    $b$ := $c_0 n + c_1$;
    $c$ := $c_2 n + c_3$;
    while ($a \neq 0$):
        $b$ := $c_4 a + c_5 b + c_6 c + c_7 n + c_8$;
        $c$ := $c_9 a + c_{10} b_{old} + c_{11} c + c_{12} n + c_{13}$;
        $a$ := $a - 1$;
    assert($b = n^3$);
    return $b$;
\end{lstlisting}
$C = \{ c_0, c_1, c_2, c_3, c_4, c_5, c_6, c_7, c_8, c_9, c_{10},
c_{11}, c_{12}, c_{13} \}$.  This is an example for a synthesis from
terminating lasso program skeleton.
The lasso is constructed from a complexity
guess ($n$ steps to termination) and an auxiliary variables $b$ and
$c$ (see also the program \texttt{squareS}).  We use the generic
template of degree $3$ over the variables $V = \{ a, b, c, n \}$.
\begin{align*}
\alpha(c_0) &= -\sqrt{8} + \frac{1}{2}, \quad \alpha(c_1) = 0, \quad
\alpha(c_2) = \frac{1}{2}(\sqrt{2} - 1), \quad \alpha(c_3) = \sqrt{2}, \\
\alpha(c_4) &= 0, \quad \alpha(c_5) = 1, \quad \alpha(c_6) = 1, \quad
\alpha(c_7) = \frac{1}{\sqrt{2}} + 1, \quad
\alpha(c_8) = 0, \\
\alpha(c_9) &= \frac{3}{2}, \quad \alpha(c_{10}) = 0, \quad
\alpha(c_{11}) = 1, \quad \alpha(c_{12}) = 1, \quad \alpha(c_{13}) =
-\sqrt{8} \\
\tilde\alpha(\Psi) &= b - n^3 + \frac{1}{2} a^2 n + \frac{1}{2} a^3 +
ac + \Big(\sqrt{2} - \frac{1}{2} \Big) a + \frac{1}{2} \Big(\sqrt{2} +
1 \Big) an - \sqrt{2} a^2 \\
\tilde\alpha(\Phi) &= a
\end{align*}
\fi

\end{document}